\documentclass[preprint,review,12pt]{elsarticle}
\usepackage[margin=.8in]{geometry}
\usepackage{xcolor}

\usepackage{graphics} % for pdf, bitmapped graphics files
\usepackage{epsfig} % for postscript graphics files
\usepackage{amsmath} % assumes amsmath package installed
\usepackage{amsthm}
\usepackage{amssymb}  % assumes amsmath package installed
\usepackage{subfigure}
\newtheorem{theorem}{Theorem}%[section]
\newtheorem{remark}{Remark}%[section]
\newtheorem{example}{Example}%[section]
\newtheorem{definition}{Definition}

\journal{NAHS}

\begin{document}

\begin{frontmatter}



\title{Integral Input-to-State Stability of Nonlinear Time-Delay Systems with Delay-Dependent Impulse Effects}
%\thanks[footnoteinfo]{This paper was not presented at any IFAC 
%meeting. Corresponding author K.~Zhang. Tel. +1 613-533-2425. }

%\author[Paestum]{Marcus Tullius Cicero}\ead{cicero@senate.ir}
\author[label]{Kexue Zhang}\ead{kexue.zhang@ucalgary.ca}
%\author[label]{}
\address[label]{Department of Mathematics and Statistics, University of Calgary, Calgary, Alberta T2N 1N4, Canada}

\begin{abstract}
This paper studies integral input-to-state stability (iISS) of nonlinear impulsive systems with time-delay in both the continuous dynamics and the impulses. Several iISS results are established by using the method of Lyapunov-Krasovskii functionals. For impulsive systems with iISS continuous dynamics and destabilizing impulses, we derive two iISS criteria that guarantee the uniform iISS of the whole system provided that the time period between two successive impulse moments is appropriately bounded from below. Then we provide an iISS result for systems with unstable continuous dynamics and stabilizing impulses. For this scenario, it is shown that the iISS properties are guaranteed if the impulses occur frequently enough. Last but not least, sufficient conditions are also obtained to guarantee the uniform iISS of the entire system over arbitrary impulse time sequences. As applications, iISS properties of a class of bilinear systems are studied in details with simulations to demonstrate the presented results.

\end{abstract}

\begin{keyword}
Impulsive systems\sep time-delay \sep integral input-to-state stability\sep delay-dependent impulses \sep Lyapunov-Krasovskii functional.


\end{keyword}

\end{frontmatter}

% \linenumbers

\section{Introduction}

Impulsive system is a hybrid dynamical system that exhibits both continuous dynamics (modeled by differential equations) and discrete dynamics (or impulses which are state jumps or resets at a sequence of discrete moments). Due to the ubiquitous of time-delay, stability of time-delay systems with impulses has been studied extensively in the literature (see, e.g., \cite{IS:2009,ZY-DX:2007,XL-QW:2007,WR-JX:2019}). It is natural to consider the time-delay effects in the impulses of a dynamical system when the state abrupt changes depend on the state at a history moment. For instance, it takes time to sample, process and transmit the impulse information in the impulsive controller which utilizes the impulses to control a dynamical system. Discrete and distributed delays are considered in the impulsive protocols in \cite{XL-KZ-WCX:2018} and \cite{XL-KZ-WCX:2019}, respectively, to guarantee the consensus of multi-agent systems. Such time-delay arises from the postponement of communications among the agents. Sampling and transmission delays are also inevitable in impulsive synchronization based secure communications (see \cite{AK-XL-XS:2009} for example). In the past decade, great progress has been made in the study of stability properties of nonlinear systems with delay-dependent impulse effects and related control problems (see, e.g., \cite{XL-KZ-WCX:2018,XL-KZ-WCX:2019,AK-XL-XS:2009,LG-DW-GW:2015,XL-KZ-WCX:2016,XL-XZ-SS:2017,XL-KZ:2018}).

The notion of input-to-state stability (ISS), introduced by Sontag in \cite{EDS:1989}, characterizes the impact of external inputs to the control system. The ISS notion roughly states that "the state must be bounded if the external inputs are uniformly bounded". Applications of ISS are now extensive, e.g., event-triggered control \cite{PT:2007}, distributed source seeking \cite{SA-WW-FZ:2018}, robustification of observers \cite{MC-ST-LZ:2018},  formation maneuver \cite{FM-FH-MB-MQ:2018}, predictive control \cite{DL-MMN-IVK:2018}. However, the boundedness of the states cannot be reflected through the ISS property when the inputs are unbounded but have finite total energy (e.g., impulse inputs). Another drawback of the ISS property is that it cannot provide an ideal bound if the external inputs have an extremely large bound but finite total energy. To take into account of such inputs, an integral variant of ISS, called integral input-to-state stability (iISS), was introduced in \cite{EDS:1998}. Applications of iISS can be found in small-gain theorem \cite{HI-CMK:2018}, tracking problems \cite{DA-EDS-YW:2000}, disturbance attenuation \cite{DL:1999}, and so on. The concepts of ISS and iISS were subsequently extended to impulsive systems in \cite{JPH-DL-ART:2008,WHC-WXZ:2009} and to hybrid impulsive and switching systems with time-delay in \cite{JL-XL-WCX:2011}. The ISS and iISS results were further improved in \cite{XMS-WW:2012} for switching time-delay systems with impulses. Up to now, numerous researchers have investigated the ISS properties of impulsive systems (see, e.g., \cite{XW-YT-WZ:2016,YW-RW-JZ:2013,SD-AM:2013}). However, no delay effects have been considered in these impulses. In \cite{XZ-XL:2016}, the effects of delay-dependent impulses were studied for the ISS properties of nonlinear delay-free systems for the first time. Recently, we investigated the ISS properties of nonlinear time-delay systems with delay-dependent impulses in \cite{XL-KZ:2019}. Nevertheless, the iISS properties of such impulsive systems have not been studied yet.

Motivated by the above discussion, we study the integral input-to-state stability of nonlinear time-delay systems with delay-dependent impulses. The stability analysis is conducted by using the method of Lyapunov-Krasovskii functionals. In general, the Lyapunov candidate is broken into a function part and a functional part in order to follow the spirit of Lyapunov and Krasovskii and characterize the impulse effects simultaneously. To be more specific, the instantaneous state jumps can be captured through the function part of the Lyapunov candidate, while the functional part is indifferent to impulses. {To our best knowledge, the time-delay effects have been only considered in either the continuous dynamics (see, \cite{WHC-WXZ:2009,JL-XL-WCX:2011,XMS-WW:2012,XW-YT-WZ:2016,YW-RW-JZ:2013}) or the impulses (see, e.g., \cite{XZ-XL:2016}) in the existing iISS results for impulsive systems, and none of the existing iISS results is applicable to impulsive systems with time-delay in both the continuous dynamics and the impulses (namely, time-delay systems with delay-dependent impulses). It is the first time that the iISS results are derived for nonlinear time-delay systems with delay-dependent impulse effects in this study. Additionally, compared with the results in \cite{JL-XL-WCX:2011,XMS-WW:2012}, one of our iISS criteria is less conservative when applied to time-delay systems with delay-free impulses (See Remark \ref{remark-th3} for details).} When the obtained results are applied to analyze the iISS properties of a type of bilinear systems, our sufficient iISS conditions generalize the results in \cite{PP-ZPJ:2006} for time-delay systems without impulses and the ones in \cite{JL-XL-WCX:2011} for time-delay systems with delay-free impulse effects, and can be used to systems with delay-dependent impulses.

The rest of this paper is organized as follows. Section \ref{Sec2} introduces some preliminaries. In section \ref{Sec3}, we present our main results for iISS of nonlinear impulsive systems with time-delay in both the continuous and discrete dynamics. The iISS properties of a type of bilinear systems are investigated in section \ref{Sec4} and two illustrative examples are then provided in section \ref{Sec5} to demonstrate the obtained iISS criteria. Section \ref{Sec6} concludes the paper.

\section{Preliminaries}\label{Sec2}

This section introduces our notation and states the iISS problem of impulsive systems with time-delay.

\subsection{Notation}
Let $\mathbb{N}$ denote the set of positive integers, $\mathbb{Z}^+$ the set of nonnegative integers, $\mathbb{R}$ the set of real numbers, $\mathbb{R}^+$ the set of nonnegative reals, and $\mathbb{R}^n$ the $n$-dimensional real space equipped with the Euclidean norm denoted by $\|\cdot\|$. Let $I$ denote the identity matrix. Given an $n\times n$ matrix $A$, we denote $\|A\|$ the spectral norm of $A$. For a positive definite matrix $P$, let $\lambda_{max}(P)$ and $\lambda_{min}(P)$ represent the largest and smallest eigenvalues of $P$, respectively.

The following function classes are essential to our study of the iISS property. A continuous function $\alpha:\mathbb{R}^+\rightarrow\mathbb{R}$ is said to be of class $\mathcal{K}$ and we write $\alpha\in\mathcal{K}$, if $\alpha$ is strictly increasing and $\alpha(0)=0$. If $\alpha$ is also unbounded, we say that $\alpha$ is of class $\mathcal{K}_{\infty}$ and we write $\alpha\in\mathcal{K}_{\infty}$. A continuous function $\beta:\mathbb{R}^+\times\mathbb{R}^+ \rightarrow\mathbb{R}^+$ is said to be of class $\mathcal{KL}$ and we write $\beta\in \mathcal{KL}$, if $\beta(\cdot,t)\in\mathcal{K}$ for each $t\in\mathbb{R}^+$ and $\beta(s,t)$ decreases to $0$ as $t\rightarrow \infty$ for each $s\in \mathbb{R}^+$.

For $a,b\in \mathbb{R}$ with $b>a$, denote $\mathcal{PC}([a,b],\mathbb{R}^n)$ the set of piecewise right continuous functions $\varphi:[a,b]\rightarrow\mathbb{R}^n$, and $\mathcal{PC}([a,\infty),\mathbb{R}^n)$ the set of functions $\phi:[a,\infty)\rightarrow\mathbb{R}^n$ satisfying $\phi|_{[a,b]}\in \mathcal{PC}([a,b],\mathbb{R}^n)$ for all $b>a$, where $\phi|_{[a,b]}$ is a restriction of $\phi$ on interval $[a,b]$. Given $r>0$, the linear space $\mathcal{PC}([-r,0],\mathbb{R}^n)$ is equipped with a norm defined by $\|\varphi\|_r:=\sup_{s\in[-r,0]}\|\varphi(s)\|$ for $\varphi\in \mathcal{PC}([-r,0],\mathbb{R}^n)$. For simplicity, we use $\mathcal{PC}$ to represent $\mathcal{PC}([-r,0],\mathbb{R}^n)$. Given $x\in \mathcal{PC}([-r,\infty),\mathbb{R}^n)$ and for each $t\in\mathbb{R}^+$, we define $x_t\in\mathcal{PC}$ as $x_t(s):=x(t+s)$ for $s\in [-r,0]$. 

\subsection{Problem Formulation}
Consider the following nonlinear time-delay impulsive system:
\begin{eqnarray}\label{sys}
\left\{\begin{array}{ll}
\dot{x}(t)=f(t,x_t,w(t)), & t\not=t_k,~k\in\mathbb{N}\cr
\Delta x(t)=I_k(t,x_{t^-},w(t^-)), & t=t_k,~k\in\mathbb{N}\cr
x_{t_0}=\varphi,
\end{array}\right.
\end{eqnarray}
where $x(t)\in\mathbb{R}^n$ is the system state; $w\in \mathcal{PC}([t_0,\infty),\mathcal{R}^m)$ is the input function; $\varphi\in\mathcal{PC}$ is the initial function; $f,I_k:\mathbb{R}^+\times\mathcal{PC}\times\mathbb{R}^m\rightarrow\mathbb{R}^n$ satisfy $f(t,0,0)=I_k(t,0,0)=0$ for all $k\in\mathbb{N}$; $\{t_1,t_2,t_3,...\}$ is a strictly increasing sequence such that $t_k\rightarrow\infty$ as $t\rightarrow\infty$; $\Delta x(t):=x(t^+)-x(t^-)$ where $x(t^+)=\lim_{s\rightarrow t^+} x(s)$ and $x(t^-)=\lim_{s\rightarrow t^-} x(s)$ (similarly, $w(t^-)=\lim_{s\rightarrow t^-} w(s)$); $x_{t^-}$ is defined as $x_{t^-}(s)=x(t+s)$ if $s\in[-r,0)$ and $x_{t^-}(0)=x(t^-)$; $r$ represents the maximum time-delay in system \eqref{sys}. Given $w\in \mathcal{PC}([t_0,\infty),\mathcal{R}^m)$, define $g(t,\phi)=f(t,\phi,w(t))$ and assume $g$ satisfies all the necessary conditions in \cite{GB-XL:1999} so that, for any initial condition $\varphi\in\mathcal{PC}$, system \eqref{sys} has a unique solution $x(t,t_0,\varphi)$ that exists in a maximal interval $[t_0-r,t_0+\Gamma)$, where $0<\Gamma\leq \infty$.

Now we state the iISS definition for system \eqref{sys} which was originally introduced for hybrid impulsive and switching systems with time-delay in \cite{JL-XL-WCX:2011}.

\begin{definition}
System \eqref{sys} is said to be uniformly integral input-to-state stable (iISS) over a certain class $\ell$ of admissible impulse time sequences, if there exist functions $\beta\in\mathcal{KL}$ and $\alpha,\gamma\in\mathcal{K}_{\infty}$, independent of the choice of the sequences in $\ell$, such that, for each initial condition $\varphi\in\mathcal{PC}$ and input function $w\in\mathcal{PC}([t_0,\infty),\mathbb{R}^m)$, the corresponding solution to \eqref{sys} exists globally and satisfies
$$\alpha(\|x(t)\|)\leq \beta(\|\varphi\|_r,t-t_0)+\int^t_{t_0} \gamma(\|w(s)\|)\mathrm{d}s + \sum_{t_0<t_k\leq t} \gamma(\|w(t^-_k)\|)$$
for all $t\geq t_0$.
\end{definition}

To study the iISS properties of system \eqref{sys}, we apply the method of Lyapunov-Krasovskii functionals and require that the Lyapunov functional candidate contains a pure function portion which can be used to characterize the impulse effects on the whole Lyapunov candidate. Next we introduce two function classes related to such function part and the Lyapunov functional candidate, respectively. 
\begin{definition}
A function $v:\mathbb{R}^+\times\mathbb{R}^n\rightarrow \mathbb{R}^+$ is said to be of class $\nu_0$ and we write $v\in\nu_0$, if, for each $x\in\mathcal{PC}(\mathbb{R}^+,\mathbb{R}^n)$, the composite function $t\mapsto v(t,x(t))$ is also in $\mathcal{PC}(\mathbb{R}^+,\mathbb{R}^n)$ and can be discontinuous at some $t'\in\mathbb{R}^+$ only when $t'$ is a discontinuity point of $x$. 
\end{definition}

\begin{definition}
A functional $v:\mathbb{R}^+\times \mathcal{PC}\rightarrow\mathbb{R}^n$ is said to be of class $\nu^*_0$ and we write $v\in \nu^*_0$, if, for each function $x\in\mathcal{PC}([-r,\infty),\mathbb{R}^n)$, the composite function $t\mapsto v(t,x_t)$ is continuous in $t$ for all $t\geq 0$. 
\end{definition}

To analyze the continuous dynamics of system \eqref{sys}, we introduce the upper right-hand derivative of the Lyapunov functional candidate $V(t,x_t)$ with respect to system \eqref{sys}:
$$\mathrm{D}^+V(t,\phi)=\limsup_{h\rightarrow 0^+}\frac{1}{h}[V(t+h,x_{t+h}(t,\phi))-V(t,\phi)],$$
where $x(t,\phi)$ is a solution to \eqref{sys} satisfying $x_t=\phi$.

\section{Sufficient Conditions for iISS}\label{Sec3}

In this section, we establish several iISS results for system \eqref{sys} over three types of impulse time sequences. We first introduce two results for iISS of system \eqref{sys} with stable continuous dynamics and destabilizing impulses on $\ell_{\textrm{inf}}(\delta)$, the class of impulse time sequences satisfying $\inf_{k\in\mathbb{N}}\{t_k-t_{k-1}\}\geq \delta$.

%(i) $\ell_{\textrm{inf}}(\delta)$, the class of impulse time sequences satisfying $\inf_{k\in\mathbb{N}}\{t_k-t_{k-1}\}\geq \delta$; (ii) $\ell_{\textrm{sup}}(\delta)$, the class of impulse time sequences satisfying $\sup_{k\in\mathbb{N}}\{t_k-t_{k-1}\}\leq \delta$; (iii) $\ell_{\textrm{all}}$, the set containing all the possible impulse time sequences. 

\begin{theorem}
\label{th1}
Assume that there exist $V_1\in \nu_0$, $V_2\in \nu^*_0$, functions $\alpha_1,\alpha_2,\alpha_3,\chi\in\mathcal{K}_{\infty}$ and constants $\mu>0$ and $\rho_1,\rho_2\geq 0$, such that, for all $t\in \mathbb{R}^+$, $x\in\mathbb{R}^n$, $y\in\mathbb{R}^m$ and $\phi\in \mathcal{PC}$, 
\begin{itemize}
\item[(i)] $\alpha_1(\|x\|)\leq V_1(t,x)\leq \alpha_2(\|x\|)$ and $0\leq V_2(t,\phi)\leq \alpha_3(\|\phi\|_{r})$;

\item[(ii)] $\mathrm{D}^+ V(t,\phi)\leq (\chi(\|w(t)\|)-\mu) V(t,\phi) + \chi(\|w(t)\|)$, where $V(t,\phi)=V_1(t,\phi(0))+V_2(t,\phi)$;

\item[(iii)] $V_1(t,\phi(0)+I_k(t,\phi,y))\leq \rho_1 V_1(t^-,\phi(0))+\rho_2 \sup_{s\in[-r,0]}\{V_1(t^-+s,\phi(s))\}+\chi(\|y\|)$;

%\item[(iv)] $\ln \rho <\mu\delta$ where $\rho:=\rho_1+\rho_2 e^{\mu r}$.
\end{itemize}
Moreover, if one of the following conditions holds:
\begin{itemize}
\item[(a)] $\rho_1\geq 1$ and $\ln \rho <\mu\delta$ where $\rho:=\rho_1+\rho_2 e^{\mu r}>1$;

\item[(b)] $\rho_1 < 1$ and there exists a positive constant $\kappa$ such that $V_2(t,\phi)\leq \kappa \sup_{s\in[-r,0]}\{V_1(t+s,\phi(s))\}$ and $\ln \rho <\mu\delta$ where $\rho:=\rho_1+[\rho_2+(1-\rho_1)\kappa] e^{\mu r}$ with $\rho_1+\rho_2\geq 1$,
\end{itemize}
then system (\ref{sys}) is uniformly iISS over $\ell_{\textrm{inf}}(\delta)$.
\end{theorem}

\begin{proof} Let's first prove this result under condition (a). There exists a small enough constant $\lambda>0$ so that $\mu>\lambda$ and $\rho e^{-(\mu-\lambda)\delta}\leq 1$. Let $x$ be a solution of \eqref{sys}, and set $v_1(t):=V_1(t,x(t))$, $v_2(t):=V_2(t,x_t)$, and $v(t):=v_1(t)+v_2(t)$. By mathematical induction, we shall prove that
\begin{align}\label{eqn-u}
v(t)e^{\lambda(t-t_0)} \leq  \mathcal{E}(t,t_0)\Big( \alpha(\|\varphi\|_r) +\rho e^{\lambda(t-t_0)} \int^t_{t_0}{\chi}(\|w(s)\|)\mathrm{d}s +\sum_{t_0<t_k\leq t} e^{\lambda(t_k-t_0)} \chi(\|w(t^-_k)\|) \Big),
\end{align}
where $\alpha(\|\varphi\|_r):=\alpha_2(\|\varphi(0)\|)+\alpha_3(\|\varphi\|_r)$ and $\mathcal{E}(t,s)=\mathrm{exp}(\int^t_{s}{\chi}(\|w(\tau)\|)\mathrm{d}\tau)$. For convenience, denote the right-hand side of \eqref{eqn-u} as $u(t)$.

By multiplying both sides of the inequality in condition (ii) with $\mathrm{exp}(\mu t-\int^t_{t_k}{\chi}(\|w(s)\|)\mathrm{d}s)$ and then integrating from $t_k$ to $t$, we obtain
\begin{equation}\label{eqn-v}
v(t)\leq \mathcal{E}(t,t_k) \Big( e^{-\mu(t-t_k)} v(t_k) +\int^t_{t_k}{\chi}(\|w(s)\|)\mathrm{d}s\Big),
\end{equation}
for $t\in [t_k,t_{k+1})$ and $k\in \mathbb{Z}^+$. In \eqref{eqn-v}, we used the fact that $\mathrm{exp}(-\mu(t-s)-\int^s_{t_k}{\chi}(\|w(\tau)\|)\mathrm{d}\tau)\leq 1$ for all $s\in[t_k,t]$. For $t\in [t_0,t_1)$, multiplying both sides of \eqref{eqn-v} with $e^{\lambda(t-t_0)}$ and using condition (i) and the fact that $\mu>\lambda$ and $\rho\geq 1$, we conclude that \eqref{eqn-u} holds for $t\in [t_0,t_1)$.

Now suppose that \eqref{eqn-u} holds for $t\in [t_0,t_m)$ where $m\geq 1$. We shall prove \eqref{eqn-u} is true on $[t_m,t_{m+1})$. To do this, we firstly conduct the following estimation:
\begin{align}\label{eqn-tk-}
\rho v(t^-_m) e^{\lambda(t_m-t_0)}
\leq  & ~\mathcal{E}(t_m,t_0)\Big(\rho e^{-(\mu-\lambda)(t_m-t_{m-1})} \alpha(\|\varphi\|_r)\cr
      &+ \rho^2 e^{-\mu(t_m-t_{m-1})} e^{\lambda(t_m-t_0)} \int^{t_{m-1}}_{t_0}{\chi}(\|w(s)\|)\mathrm{d}s\cr
      &+ \rho e^{-(\mu-\lambda)(t_m-t_{m-1})}  \sum_{t_0<t_k\leq t_{m-1}} e^{\lambda(t_k-t_0)}\chi(\|w(t^-_k)\|)\Big)\cr
      & + \mathcal{E}(t_m,t_{m-1}) \rho e^{\lambda(t_m-t_0)} \int^{t_{m}}_{t_{m-1}}{\chi}(\|w(s)\|)\mathrm{d}s\cr
\leq  & ~\mathcal{E}(t_m,t_0)\Big(\rho e^{-(\mu-\lambda)\delta} \alpha(\|\varphi\|_r)\cr
      &+ \rho e^{\lambda(t_m-t_0)} \int^{t_{m}}_{t_0}{\chi}(\|w(s)\|)\mathrm{d}s\cr
      &+ \rho e^{-(\mu-\lambda)\delta}  \sum_{t_0<t_k\leq t_{m-1}} e^{\lambda(t_k-t_0)}\chi(\|w(t^-_k)\|)\Big)\cr
\leq  & ~u(t^-_m).
\end{align}
We used \eqref{eqn-v} with $t=t^-_m$ and then \eqref{eqn-u} with $t=t_{m-1}$ in the first inequality of \eqref{eqn-tk-}. For the second inequality of \eqref{eqn-tk-}, we used the facts that $t_m-t_{m-1}\geq \delta$, $\rho e^{-\mu\delta}< 1$ and $\mathcal{E}(t_m,t_{m-1})\leq \mathcal{E}(t_m,t_{0})$. Next, we will show that \eqref{eqn-u} is true for $t=t_m$. We start with making the claim that
\begin{equation}\label{eqn-claim}
\rho v(t^-_m+s)e^{\lambda(t_m+s-t_0)}\leq e^{(\mu-\lambda)r} u(t^-_m), \textrm{~for~all~} s\in[-r,0].
\end{equation}
Without loss of generality, we assume $t_m+s\geq t_0$ for all $s\in[-r,0]$, then, for a fixed $s\in[-r,0]$, there exists an integer $j$ ($0\leq j\leq m-1$) such that $t_m+s\in [t_j,t_{j+1})$. By using \eqref{eqn-v} with $t=t^-_m+s$ and then \eqref{eqn-u} with $t=t_j$, we obtain:
\begin{align}
& \rho v(t^-_m+s) e^{\lambda(t_m+s-t_0)}\cr
\leq &  {~\rho\mathcal{E}(t_m+s,t_{j}) \Big\{ e^{(\lambda-\mu)(t_m+s-t_0)} \Big[ \mathcal{E}(t_j,t_{0}) \Big( \alpha(\|\varphi\|_r)+\rho e^{\lambda(t_{j}-t_0)} \int^{t_{j}}_{t_0}{\chi}(\|w(s)\|)\mathrm{d}s }\cr
     &  {~~~~ +\sum_{t_0<t_k\leq t_{j}} e^{\lambda(t_k-t_0)}\chi(\|w(t^-_k)\|)\Big)\Big] + e^{\lambda(t_m+s-t_0)}\int^{t_{m}+s}_{t_j}{\chi}(\|w(s)\|)\mathrm{d}s \Big\}}\cr
\leq & ~\mathcal{E}(t_m+s,t_{j})\Big\{ e^{(\mu-\lambda)r}\mathcal{E}(t_j,t_{0})\Big[ \rho e^{-(\mu-\lambda)\delta} \alpha(\|\varphi\|_r)+\rho^2 e^{-\mu\delta} e^{\lambda(t_{j+1}-t_0)} \int^{t_{j}}_{t_0}{\chi}(\|w(s)\|)\mathrm{d}s \cr
     & ~~~~ + \rho e^{-(\mu-\lambda)\delta}  \sum_{t_0<t_k\leq t_{j}} e^{\lambda(t_k-t_0)}\chi(\|w(t^-_k)\|)\Big] + \rho e^{\lambda(t_{j+1}-t_0)} \int^{t_{j+1}}_{t_j}{\chi}(\|w(s)\|)\mathrm{d}s \Big\}\cr
\leq & e^{(\mu-\lambda)r} u(t^-_{j+1})\cr
\leq & e^{(\mu-\lambda)r} u(t^-_m),
\end{align}
which implies \eqref{eqn-claim} is true for all $s\in[-r,0]$.  {In the second inequality, we used the facts $t_{j+1}-t_m-s\leq r$ and $e^{(\lambda-\mu)(t_m+s-t_0)}=e^{(\mu-\lambda)(t_{j+1}-t_m-s)}e^{(\lambda-\mu)(t_{j+1}-t_j)}$. The last inequality is from the fact that $u$ is a monotone increasing function.} Combining \eqref{eqn-tk-} and \eqref{eqn-claim} with condition (iii) and the fact that $\rho_1\geq 1$, we conclude that
\begin{align}\label{eqn-tm}
 v(t_m) e^{\lambda(t_m-t_0)}
\leq & ~[\rho_1 v_1(t^-_m) + \rho_2 \sup_{-r\leq s\leq 0}\{v_1(t^-_m+s)\} + \chi(\|w(t^-_m)\|) + v_2(t^-_m)] e^{\lambda(t_m-t_0)}\cr
\leq & ~[\rho_1 v(t^-_m) + \rho_2 \sup_{-r\leq s\leq 0}\{v(t^-_m+s)\} + \chi(\|w(t^-_m)\|)] e^{\lambda(t_m-t_0)}\cr
\leq & ~\rho_1 v(t^-_m) e^{\lambda(t_m-t_0)} + \rho_2 e^{\lambda r} \sup_{-r\leq s\leq 0}\{v(t^-_m+s) e^{\lambda(t_m+s-t_0)}\} + \chi(\|w(t^-_m)\|) e^{\lambda(t_m-t_0)}\cr
\leq & ~\frac{\rho_1+\rho_2 e^{\mu r}}{\rho} u(t^-_m) + \chi(\|w(t^-_m)\|) e^{\lambda(t_m-t_0)}\cr
  =  & ~u(t_m),
\end{align}
which implies \eqref{eqn-u} holds for $t=t_m$. {In the first inequality, we used the fact that $v_2$ is continuous at $t=t_m$.} We now prove that \eqref{eqn-u} is true on $(t_m,t_{m+1})$. For  $t\in(t_m,t_{m+1})$, we have
\begin{align}\label{eqn-tm_tm+1}
 v(t) e^{\lambda(t-t_0)}
\leq & ~\mathcal{E}(t,t_{0}) \bigg( e^{-(\mu-\lambda)(t-t_m)}\alpha(\|\varphi\|_r) +\rho e^{-\mu(t-t_m)}e^{\lambda(t-t_0)}\int^{t_m}_{t_0}\chi(\|w(s)\|)\mathrm{d}s\cr
     & ~~+e^{-(\mu-\lambda)(t-t_m)} \sum_{t_0<t_k\leq t_m}e^{\lambda(t_k-t_0)}\chi(\|w(t^-_k)\|) \bigg)\cr
     & ~~+\mathcal{E}(t,t_{m})e^{\lambda(t-t_0)} \int^{t}_{t_m}\chi(\|w(s)\|)\mathrm{d}s\cr
\leq & ~\mathcal{E}(t,t_{0}) \bigg( \alpha(\|\varphi\|_r) +\rho e^{\lambda(t-t_0)}\int^{t_m}_{t_0}\chi(\|w(s)\|)\mathrm{d}s+ \sum_{t_0<t_k\leq t_m}e^{\lambda(t_k-t_0)}\chi(\|w(t^-_k)\|) \bigg)\cr
     & ~~+\mathcal{E}(t,t_{0})e^{\lambda(t-t_0)} \int^{t}_{t_m}\chi(\|w(s)\|)\mathrm{d}s\cr
\leq & u(t).
\end{align}
Here, we used \eqref{eqn-v} and then \eqref{eqn-u} with $t=t_m$ for the first inequality in \eqref{eqn-tm_tm+1}. For the second inequality, we used the facts that $\mu>\lambda$ and $\mathcal{E}(t,t_{0})\geq \mathcal{E}(t,t_{m})$. The last inequality in \eqref{eqn-tm_tm+1} holds because $\rho\geq 1$.

By induction, we conclude that \eqref{eqn-u} is true for all $t\geq t_0$. The iISS estimation can be conducted from \eqref{eqn-u} by standard arguments. The details are essentially the same as that in Theorem 3.3 of \cite{JL-XL-WCX:2011} and thus omitted. Boundedness of the solution to \eqref{sys} follows from this estimate, which then implies the solution's global existence (see the continuation theorem in \cite{GB-XL:1999}).

With condition (b), the proof is identical to the above discussion. The main difference is to replaced the following estimate of $v(t_m)$ in \eqref{eqn-tm}:
\begin{align}
 v(t_m) 
 \leq & \rho_1 v_1(t^-_m) + \rho_2 \sup_{-r\leq s\leq 0}\{v_1(t^-_m+s)\} + \chi(\|w(t^-_m)\|) + (\rho_1 +1-\rho_1)v_2(t^-_m)\cr
\leq & \rho_1 v(t^-_m) + [\rho_2+(1-\rho_1)\kappa] \sup_{-r\leq s\leq 0}\{v(t^-_m+s)\} + \chi(\|w(t^-_m)\|),\nonumber
\end{align}
where we used the condition $v_2(t)\leq \kappa \sup_{s\in[-r,0]}\{v_1(t+s)\}$. The rest of the proof is omitted.
\end{proof}

\begin{remark}\label{remark-th1} 
Compared with the existing results in \cite{JL-XL-WCX:2011,XMS-WW:2012,XW-YT-WZ:2016,YW-RW-JZ:2013,SD-AM:2013}, the main contribution of Theorem \ref{th1} is that it can be used to analyze the iISS properties of nonlinear systems with delay-dependent impulses. Condition (iii) describes the impulse effects on the function portion of the Lyapunov functional candidate. Parameters $\rho_1$ and $\rho_2$ quantify these effects related to the non-delayed and delayed states, respectively. If $\rho_2=0$, then Theorem \ref{th1} reduces to a special case of Theorem 3.3 in \cite{JL-XL-WCX:2011} for switching-free system \eqref{sys}. In condition (a), $\rho_1\geq 1$ means the non-delayed states at each impulse moment play the key role in the destabilizing impulse effects. On the other hand, $\rho_1<1$ in condition (b) implies that not the non-delayed states but the delayed states (or together with the non-delayed ones) lead to the destabilizing impulse effects. {Intuitively, conditions (a) and (b) say that increasing $\rho_1$ or $\rho_2$ corresponds to enlarging the destabilizing influence of the impulses which leads to a bigger $\delta$, and then the impulses cannot occur too frequently so that the entire system is iISS.}
\end{remark}

When the impulses are stabilizing but the continuous dynamics is unstable, we introduce an iISS criterion for system \eqref{sys} over $\ell_{\textrm{sup}}(\delta)$, the class of impulse time sequences satisfying $\sup_{k\in\mathbb{N}}\{t_k-t_{k-1}\}\leq \delta$.

\begin{theorem}
\label{th3}
Assume that there exist $V_1\in \nu_0$, $V_2\in \nu^*_0$, functions $\alpha_1,\alpha_2,\alpha_3,\chi\in\mathcal{K}_{\infty}$ and constants $\mu>0$, $\kappa>0$, $1>\rho_1\geq 0$ and $\rho_2\geq 0$, such that, for all $t\in \mathbb{R}^+$, $x\in\mathbb{R}^n$, $y\in\mathbb{R}^m$, and $\phi\in \mathcal{PC}$, 
\begin{itemize}
\item[(i)] $\alpha_1(\|x\|)\leq V_1(t,x)\leq \alpha_2(\|x\|)$ and $0\leq V_2(t,\phi)\leq \alpha_3(\|\phi\|_{r})$;

\item[(ii)] $\mathrm{D}^+ V(t,\phi)\leq (\chi(\|w(t)\|)+\mu) V(t,\phi) + \chi(\|w(t)\|)$, where $V(t,\phi)=V_1(t,\phi(0))+V_2(t,\phi)$;

\item[(iii)] $V_1(t,\phi(0)+I_k(t,\phi,y))\leq \rho_1 V_1(t^-,\phi(0))+\rho_2 \sup_{s\in[-r,0]}\{V_1(t^-+s,\phi(s))\}+ \chi(\|y\|)$;

\item[(iv)] $V_2(t,\phi)\leq \kappa \sup_{s\in[-r,0]}\{V_1(t+s,\phi(s))\}$;

\item[(v)] $\ln[\rho_1+\rho_2+(1-\rho_1)\kappa]  < -\mu \delta$,
\end{itemize}
then system (\ref{sys}) is uniformly iISS over $\ell_{\textrm{sup}}(\delta)$.
\end{theorem}

\begin{proof}
We conclude from condition (v) that there exists a positive constant $\lambda$ close to zero so that $\ln(\rho_1+[\rho_2+(1-\rho_1)\kappa] e^{\lambda r}) \leq -(\mu+\lambda)\delta$. Denote $\rho:=\rho_1+[\rho_2+(1-\rho_1)\kappa] e^{\lambda r}$, then we have $\rho e^{(\mu+\lambda)\delta}\leq 1$. Let $M=e^{(\mu+\lambda)\delta}$ and $c=e^{\mu\delta}$, then we shall show that 
\begin{align}\label{eqn-u2}
v(t)e^{\lambda(t-t_0)} \leq  \mathcal{E}(t,t_0)\bigg( M\alpha(\|\varphi\|_r) +c e^{\lambda(t-t_0)} \int^t_{t_0}{\chi}(\|w(s)\|)\mathrm{d}s +M \sum_{t_0<t_k\leq t} e^{\lambda(t_k-t_0)} \chi(\|w(t^-_k)\|) \bigg),
\end{align}
where $\alpha$ and $\mathcal{E}$ are the same as those in the proof of Theorem \ref{th1}. Let $u(t)$ represent the right-hand side of \eqref{eqn-u2}. Similar to the estimate of \eqref{eqn-v}, we can conclude from condition (ii) that 
\begin{equation}\label{eqn-v2}
v(t)\leq \mathcal{E}(t,t_k) \Big( e^{\mu(t-t_k)} v(t_k) +c \int^t_{t_k}{\chi}(\|w(s)\|)\mathrm{d}s\Big),
\end{equation}
for $t\in [t_k,t_{k+1})$ and $k\in \mathbb{Z}^+$. Then, using \eqref{eqn-v2} on $[t_0,t_1)$, we have
\begin{align}\label{eqn-k=1}
 v(t)e^{\lambda(t-t_0)}
&\leq  \mathcal{E}(t,t_0) \Big( e^{\lambda(t-t_0)}e^{\mu(t-t_0)} v(t_0) +c e^{\lambda(t-t_0)}\int^t_{t_0}{\chi}(\|w(s)\|)\mathrm{d}s\Big)\cr
&\leq  \mathcal{E}(t,t_0) \Big( e^{(\lambda+\mu)\delta} v(t_0) +c e^{\lambda(t-t_0)}\int^t_{t_0}{\chi}(\|w(s)\|)\mathrm{d}s\Big)\cr
& = u(t),
\end{align}
which means \eqref{eqn-u2} holds on $[t_0,t_1)$. Now suppose \eqref{eqn-u2} is true for $t\in [t_0,t_m)$ with $m\geq 1$. We shall prove that \eqref{eqn-u2} holds on $[t_m,t_{m+1})$. To do this, we start with proving that \eqref{eqn-u2} is true for $t=t_m$:
\begin{align}\label{eqn-tm-2}
 v(t_m)e^{\lambda(t_m-t_0)} 
\leq & \Big( \rho_1 v_1(t^-_m) + \rho_2\sup_{s\in[-r,0]}\{v_1(t^-_m+s)\} + \chi(\|w(t^-_m)\|) + {v_2(t^-_m)}\Big)e^{\lambda(t_m-t_0)} \cr
\leq & \Big(\rho_1 v(t^-_m) + [\rho_2+(1-\rho_1)\kappa]\sup_{s\in[-r,0]}\{v_1(t^-_m+s)\} + \chi(\|w(t^-_m)\|)  \Big)e^{\lambda(t_m-t_0)} \cr
\leq & \rho_1 u(t^-_m) + [\rho_2+(1-\rho_1)\kappa] e^{\lambda r}\sup_{s\in[-r,0]}\{v(t^-_m+s) e^{\lambda(t_m+s-t_0)}\} + e^{\lambda(t_m-t_0)}\chi(\|w(t^-_m)\|)  \cr
\leq & \rho u(t^-_m) + e^{\lambda(t_m-t_0)}\chi(\|w(t^-_m)\|) \cr
\leq & u(t_m).
\end{align}
Here, we used conditions (iii) and (iv) in the first inequality of \eqref{eqn-tm-2}, and then \eqref{eqn-u2} with the fact that $u(t^-_m+s)\leq u(t^-_m)$ for all $s\in[-r,0]$ in the estimate of the fourth inequality. For $t\in(t_m,t_{m+1})$, we conclude from \eqref{eqn-v2} and the fourth inequality of \eqref{eqn-tm-2} that
\begin{align}\label{eqn-vt}
 v(t)e^{\lambda(t-t_0)}
 \leq  &\mathcal{E}(t,t_m) \Big(\rho e^{(\mu+\lambda)(t-t_m)} u(t^-_m) + e^{(\mu+\lambda)(t-t_m)}e^{\lambda(t_m-t_0)}\chi(\|w(t^-_m)\|) \cr
 & +c e^{\lambda(t-t_0)}\int^t_{t_m}\chi(\|w(s)\|)\mathrm{d}s  \Big)\cr
 \leq & \mathcal{E}(t,t_0) \Big( M \alpha(\|\varphi\|_r) +  c e^{\lambda(t-t_0)} \int^{t_m}_{t_0}\chi(\|w(s)\|)\mathrm{d}s + M \sum_{t_0<t<t_m}e^{\lambda(t_k-t_0)}\chi(\|w(t^-_k)\|) \Big)\cr
      & +\mathcal{E}(t,t_m) \Big( M e^{\lambda(t_m-t_0)} \chi(\|w(t^-_m)\|) + c e^{\lambda(t-t_0)} \int^{t}_{t_m}\chi(\|w(s)\|)\mathrm{d}s \Big) \cr
 \leq & u(t)
\end{align}
i.e., \eqref{eqn-u2} hold on $(t_m,t_{m+1})$. We used the facts that $\rho e^{(\mu+\lambda)\delta}\leq 1$ and $\mathcal{E}(t,t_m)\leq \mathcal{E}(t,t_0)$ in the second inequality of \eqref{eqn-vt}. By the method of induction, we conclude that \eqref{eqn-u2} is true for all $t\geq t_0$. The iISS estimate from \eqref{eqn-u2} is similar to that in the proof of Theorem \ref{th1}, and global existence of the solution to \eqref{sys} follows from this estimate.
\end{proof}

\begin{remark}\label{remark-th3}
Compared with the iISS results in \cite{JL-XL-WCX:2011,XMS-WW:2012}, Theorem \ref{th3} is applicable to systems with delay-dependent impulses. Furthermore, when $\rho_2=0$ in condition (iii), Theorem \ref{th3} provides a less conservative iISS result for system \eqref{sys} {with delay-free impulses}. Because the upper bound of $\delta$, $\frac{-\ln(\rho_1+(1-\rho_1)\kappa)}{\mu}$, is bigger than $\frac{-\ln(\rho_1+\kappa)}{\mu}$ required in both Theorem 3.4 in \cite{JL-XL-WCX:2011} and Theorem 2 in \cite{XMS-WW:2012}, that is, our result can be applied to iISS analysis of system \eqref{sys} over a wider class of impulse sequences.
\end{remark}

{\begin{remark}\label{remark-th3.1}
It can be derived from condition (v) that $\rho_1+\rho_2+(1-\rho_1)\kappa<1$ and then $\kappa<\frac{1-\rho_1-\rho_2}{1-\rho_1}\leq 1$. Hence, condition (v) implies that decreasing $\rho_1$ or $\rho_2$ (enhancing the stabilizing impulse impact) causes a larger upper bound of $\delta$. This means the impulses should occur frequently enough to overcome the destabilizing effects of the continuous dynamics of system \eqref{sys} in order to guarantee the overall system's iISS property.
\end{remark}}

{The last result provides sufficient conditions to guarantee the uniform iISS of system \eqref{sys} over $\ell_{\textrm{all}}$, the class of arbitrary impulse sequences.}
\begin{theorem}
\label{th4}
Assume that there exist $V_1\in \nu_0$, $V_2\in \nu^*_0$, functions $\alpha_1,\alpha_2,\alpha_3,\chi\in\mathcal{K}_{\infty}$ and constants $\mu\geq 0$, $1\geq \rho_1\geq 0$ and $\rho_2\geq 0$, such that, for all $t\in \mathbb{R}^+$, $x\in\mathbb{R}^n$, $y\in\mathbb{R}^m$ and $\phi\in \mathcal{PC}$, conditions (i), (ii) and (iii) of Theorem \ref{th1} are satisfied. Moreover, if one of the following conditions holds: 
\begin{itemize}
\item[(a)] $\mu>0$ and $\rho_1+\rho_2<1$ (or $\rho_1= 1$ with $\rho_2=0$);
\item[(b)] $\mu=0$ and there exist a constant $\kappa>0$ such that condition (iv) of Theorem \ref{th3} holds and $\rho_1+\rho_2+(1-\rho_1)\kappa <1$,
\end{itemize}
then system (\ref{sys}) is uniformly iISS over $\ell_{\textrm{all}}$.
\end{theorem}

\begin{proof}

Let us first suppose condition (a) holds, then there exists a small enough constant $\lambda>0$ such that $\mu>\lambda$ and $\rho_1+\rho_2 e^{\lambda r}\leq 1$. Similar to the proof of Theorem \ref{th1}, we use mathematical induction to show \eqref{eqn-u} is true {with $\rho=1$}. The main difference lies in the estimate of \eqref{eqn-tm}:
\begin{align}
 v(t_m) e^{\lambda(t_m-t_0)}
\leq & ~\rho_1 v(t^-_m) e^{\lambda(t_m-t_0)} + \rho_2 e^{\lambda r} \sup_{-r\leq s\leq 0}\{v(t^-_m+s) e^{\lambda(t_m+s-t_0)}\} + \chi(\|w(t^-_m)\|) e^{\lambda(t_m-t_0)}\cr
  =  & ~\rho_1 u(t^-_m)+\rho_2 e^{\lambda r} \sup_{-r\leq s\leq 0}\{u(t^-_m+s)\} + \chi(\|w(t^-_m)\|) e^{\lambda(t_m-t_0)}\cr
\leq & ~(\rho_1+\rho_2 e^{\lambda r}) u(t^-_m) + \chi(\|w(t^-_m)\|) e^{\lambda(t_m-t_0)}\cr
\leq & ~u(t_m).\nonumber
\end{align}
The rest of the proof is identical to that of Theorem \ref{th1} and thus omitted.

If condition (b) holds, the result can be derived by letting $\mu$ go to zero in Theorem \ref{th3}.
\end{proof}

{
\begin{remark}\label{remark3}
It can be derived from condition (iii) of the above mentioned results that
\begin{align}\label{cond3}
V_1(t,\phi(0)+I_k(t,\phi,y))\leq& \rho_1 V_1(t^-,\phi(0))+\rho_2 \sup_{s\in[-r,0]}\{V_1(t^-+s,\phi(s))\}+\chi(\|y\|)\cr
\leq & \bar{\rho}\sup_{s\in[-r,0]}\{V_1(t^-+s,\phi(s))\}+\chi(\|y\|)
\end{align}
with $\bar{\rho}=\rho_1+\rho_2$. Therefore, letting $\rho_1=0$ is equivalent to replacing the inequality in condition (iii), $\rho_1$ and $\rho_2$ with inequality \eqref{cond3}, $0$ and $\bar{\rho}$, respectively, in our obtained results. However, condition (b) of both Theorems \ref{th1} and \ref{th3} and condition (v) of Theorem \ref{th4} are more conservative with such replacement. If $\rho_2=0$ in condition (iii) of Theorem \ref{th1}, it can be seen that $\rho_1>1$ implies the impulses can destabilize the overall system while $\rho_1<1$ means the impulses are potentially stabilizing. However, the interpretations of parameter $\rho_2$ are quite different. If $\rho_2\geq 1$, then we can derive from condition (iii) of Theorem \ref{th1} that $v_1(t_k)$ can be larger than $v_1(t^-_k)$ since $\sup_{s\in[-r,0]}\{v_1(t^-_k+s)\}\geq v_1(t^-_k)$, and then the impulses can destroy the stability of the entire system. On the other hand, if $\rho_2<1$, it is still possible for $v_1(t_k)$ to be bigger than $v(t^-_k)$ according to condition (iii) of Theorem \ref{th1}. To be more specific, $\rho_2\sup_{s\in[-r,0]}\{v_1(t^-_k+s)\}$ can be larger than $v(t^-_k)$ due to the existence of time-delay in the impulses. Therefore, the impulses can be destabilizing even $\rho_2<1$ (see Example \ref{eg1} for a demonstration with numerical simulations). Based on the above discussion, we can see that condition (a) of Theorem \ref{th4} allows the continuous dynamics of system \eqref{sys} is iISS with the external input $w$ and exponentially stable without $w$, and the impulses either contribute to the iISS of system \eqref{sys} or are destabilizing but not destroy the iISS of the overall system. Condition (b) of Theorem \ref{th4} requires the continuous dynamics is marginally stable and the impulses mainly contribute to the iISS of the entire system. 
\end{remark}

\begin{remark}\label{remark4}
Although the iISS properties of system \eqref{sys} have been investigated in \cite{WHC-WXZ:2009,JL-XL-WCX:2011,XMS-WW:2012}, the obtained condition $V_1(t_k,x(t^-_k)+I_k(t_k,x_{t^-_k},w(t^-_k)))\leq \rho V_1(t^-_k,x(t^-_k))$ at each impulse time normally can not be verified for time-delay systems with delay-dependent impulses. Therefore, the results in \cite{WHC-WXZ:2009,JL-XL-WCX:2011,XMS-WW:2012} are only applicable to time-delay systems with delay-free impulses. Our conditions on $V_1$ have taken into account of the time-delay effects in each impulse, so that the iISS properties can be studied for systems with delay-dependent impulses. It is worthwhile to mention that the ISS properties have been studied in \cite{SD-MK-AM-LN:2012} for a type of time-delay systems with delayed impulses which requires an explicit relation between $x_t$ and $x_{t^-}$ at each impulse time so that the difference between $v(t_k)$ and $v(t^-_k)$ can be quantified. However, the results obtained in \cite{SD-MK-AM-LN:2012} are not applicable to system \eqref{sys}, and the corresponding analysis cannot be generalized to investigate the iISS properties of system \eqref{sys} mainly because such an explicit relation between $x_{t_k}$ and $x_{t^-_k}$ cannot be derived from the impulses of system \eqref{sys}.
\end{remark}
}

\section{iISS of a Class of Bilinear Systems}\label{Sec4}

In this section, we use the obtained results to investigate the iISS properties of the following bilinear system:
\begin{eqnarray}\label{sys1}
\left\{\begin{array}{ll}
\dot{x}(t)=A x(t) +\sum^q_{i=1}w_i(t)(A_i x(t) +B_i x(t-r)) + C w(t), & t\not=t_k,~k\in\mathbb{N}\cr
\Delta x(t)=D x(t^-) +E x(t-d)+ F w(t^-), & t=t_k,~k\in\mathbb{N}
\end{array}\right.
\end{eqnarray}
where $A,A_i,B_i$ $(i=1,...,q)$, $D,E$ are $n\times n$ matrices, $C,F$ are $n\times q$ matrices, and $r,d$ are the delays in the continuous dynamics and the impulses, respectively. The external input $w$ is in $\mathbb{R}^q$ and its components are $w_i$ $(i=1,...,q)$, that is, $w=(w_1,w_2,...,w_q)^T$.

First, if $A$ is Hurwitz, we construct the following iISS criterion for system \eqref{sys1} from Theorem \ref{th1} and Theorem \ref{th4}.
\begin{theorem}\label{th4.1}
Suppose $A$ is Hurwitz and $P$ is the positive matrix such that $A^TP+PA=-I$. Let
$$a=\frac{\lambda_{max}((I+D)^TP(I+D))}{\lambda_{min}(P)},~b=\frac{\lambda_{max}(E^TPE)}{\lambda_{min}(P)} ~\textrm{and}~ \mu_0=\min \Big\{\frac{1}{\lambda_{max}(P)},\frac{1}{2r} \Big\}.$$
\begin{itemize}
\item[(i)] If $\sqrt{a}+\sqrt{b}>1$ (or $\sqrt{a}+\sqrt{b}=1$ with $b\not=0$) and
\begin{equation}\label{eqn4.1}
2\ln(\sqrt{a}+\sqrt{be^{\mu_0 r}})<\mu_0 \delta,
\end{equation}
then system \eqref{sys1} is uniformly iISS over $\ell_{\textrm{inf}}(\delta)$.

\item[(ii)] If $\sqrt{a}+\sqrt{b}<1$ (or $a=1$ with $b=0$), then system \eqref{sys1} is uniformly iISS over $\ell_{\textrm{all}}$.
\end{itemize}
\end{theorem}

\begin{proof}
We first prove (i). By \eqref{eqn4.1}, we can find a small enough $\xi>0$ so that 
$$2\ln(\sqrt{a}+\sqrt{(1+\xi)be^{\mu_0 r}})<\mu_0 \delta,$$
and then there exists a positive $\varepsilon$ close to zero such that $\varepsilon<1/3$ and
\begin{equation}\label{eqn4.1.1}
\ln([\sqrt{a}+\sqrt{(1+\xi)be^{\mu r}}]^2 + \kappa e^{\mu r})<\mu \delta,
\end{equation}
where
$$\mu=\min \Big\{\frac{1-3\varepsilon}{\lambda_{max}(P)},\frac{1}{2r} \Big\} ~\textrm{and}~ \kappa=\frac{3r\varepsilon}{2\lambda_{min}(P)}.$$
Let
$$\epsilon=\sqrt{\frac{(1+\xi) b e^{\mu r}}{a}},$$
then denote
$$\rho_1=(1+\epsilon)a,~~\rho_2=(1+\epsilon^{-1})(1+\xi)b.$$
We conclude from \eqref{eqn4.1.1} that
\begin{itemize}
\item if $\rho_1<1$, then
\begin{equation}\label{eqn4.1.2}
\ln(\rho_1+[\rho_2 +(1-\rho_1)\kappa] e^{\mu r}) <\mu \delta;
\end{equation}
\item if $\rho_1\geq 1$, then
\begin{equation}\label{eqn4.1.3}
\ln(\rho_1+\rho_2 e^{\mu r}) <\mu \delta.
\end{equation}
\end{itemize}

Consider the Lyapunov-Krasovskii functional $V(t)=V_1(t)+V_2(t)$ with 
$$V_1(t)=x^TPx,~~V_2(t)=\varepsilon\int^t_{t-r}\big(2+\frac{s-t}{r}\big)x^T(s)x(s)\mathrm{d}s,$$
then we can see that condition (i) of Theorem \ref{th1} is satisfied with $\alpha_1(\|x\|)=\lambda_{min}(P)\|x\|^2$, $\alpha_2(\|x\|)=\lambda_{max}(P)\|x\|^2$ and $\alpha_3(\|\phi\|_r)=2\varepsilon r\|\phi\|^2_r$.

Let
$$p_1=\max_{i=1,...,q}\{\|PA_i\|\}~\textrm{and}~p_2=\max_{i=1,...,q}\{\|PB_i\|\}.$$
From the continuous dynamics of system \eqref{sys1}, we obtain
\begin{align}
\dot{V_1}(t) &= 2x^TPAx + \sum^q_{i=1}w_i(2x^TPA_ix +2x^TPB_ix(t-r)) +2x^T PCw\cr
             &\leq x^T(A^TP+PA)x +2qp_1\|w\|\|x\|^2 + 2qp_2\|w\|\|x\|\|x(t-r)\| +2x^T PCw\cr 
             &\leq x^T(A^TP+PA+\varepsilon I)x +(2qp_1\|w\|+\varepsilon^{-1} q^2p^2_2\|w\|^2)\|x\|^2 + \varepsilon\|x(t-r)\|^2 \cr
             &~~~~+\varepsilon^{-1}\|PC\|^2\|w\|^2\cr 
\dot{V_2}(t) &= 2\varepsilon \|x\|^2 -\varepsilon\|x(t-r)\|^2-\frac{\varepsilon}{r}\int^t_{t-r}x^T(s)x(s)\mathrm{d}s, \nonumber           
\end{align}
then
\begin{align}\label{eqn4.1.4}
\dot{V}(t) &\leq x^T(A^TP+PA+3\varepsilon I)x +(2qp_1\|w\|+\varepsilon^{-1} q^2p^2_2\|w\|^2)\|x\|^2 -\frac{\varepsilon}{r}\int^t_{t-r}x^T(s)x(s)\mathrm{d}s \cr
             &~~~~+\varepsilon^{-1}\|PC\|^2\|w\|^2\cr 
           &\leq \Big(\chi_1(\|w\|)-\frac{1-3\varepsilon}{\lambda_{max}(P)} \Big)V_1(t)-\frac{1}{2r}V_2(t) +\chi_2(\|w\|)  \cr
           &\leq (\chi_1(\|w\|)-\mu )V(t) +\chi_2(\|w\|),
\end{align}
where
$$\chi_1(\|w\|)=\frac{2qp_1\|w\|+\varepsilon^{-1} q^2p^2_2\|w\|^2}{\lambda_{min}(P)}~\textrm{and}~\chi_2(\|w\|)=\varepsilon^{-1}\|PC\|^2\|w\|^2.$$

From the impulse effects of system \eqref{sys1}, we have
\begin{align}\label{eqn4.1.5}
V_1(t_k)&\leq (1+\epsilon)x^T(t^-_k)(I+D)^TP(I+D)x(t^-_k) \cr
        &~~~~+(1+\epsilon^{-1})[(1+\xi)x^T(t_k-d)E^TPEx(t_k-d)+(1+\xi^{-1})w^T(t^-_k)F^TPFw(t^-_k)]\cr
        &\leq \rho_1 V_1(t^-_k) +\rho_2 V_1(t_k-d) +\chi_3(\|w(t^-_k)\|)
\end{align}
where
$$\chi_3(\|w(t^-_k)\|)=(1+\epsilon^{-1})(1+\xi^{-1}){\lambda_{max}(F^TPF)}\|w(t^-_k)\|^2.$$
Furthermore, we have
\begin{align}\label{eqn4.1.6}
V_2(t)&\leq \varepsilon \sup_{s\in[-r,0]}\{x^T(t+s)x(t+s)\}\int^t_{t-r}2+\frac{s-t}{r}\mathrm{d}s\cr
      &\leq \frac{3}{2}r\varepsilon \sup_{s\in[-r,0]}\{x^T(t+s)x(t+s)\}\cr
      &\leq \kappa \sup_{s\in[-r,0]}\{V_1(t+s)\}.
\end{align}

We conclude from \eqref{eqn4.1.4},\eqref{eqn4.1.5},\eqref{eqn4.1.3} that conditions (ii),(iii),(a) of Theorem \ref{th1} hold with $\chi=\max\{\chi_1,\chi_2,\chi_3\}$ {and the maximum time-delay $\tau=\max\{r,d\}$}, while \eqref{eqn4.1.2} implies condition (b) of Theorem \ref{th1} is true. Therefore, system \eqref{sys1} is uniformly iISS over $\ell_{\textrm{inf}}(\delta)$.

To prove (ii), we can conclude directly from Theorem \ref{th4}.
\end{proof}

Next, we derive the following result from Theorem \ref{th3} when $A$ is not a Hurwitz matrix.

\begin{theorem}\label{th4.2}
If $A$ is not Hurwitz and
\begin{equation}\label{eqn4.2}
2\ln(\|I+D\|+\|E\|) < -\lambda_{max}(A+A^T) \delta,
\end{equation}
then system \eqref{sys1} is uniformly iISS over $\ell_{\textrm{sup}}(\delta)$.
\end{theorem}

\begin{proof}
It can be seen from \eqref{eqn4.2} that there exists a small enough $\xi>0$ so that 
$$2\ln(\|I+D\|+\sqrt{1+\xi}\|E\|) < -\lambda_{max}(A+A^T) \delta.$$
Let $\epsilon=\frac{\sqrt{1+\xi}\|E\|}{\|I+D\|}$, then we can find a positive $\varepsilon$ close to zero such that
\begin{equation}\label{eqn4.2.1}
\ln(\rho_1+\rho_2+(1-\rho_1)\kappa)<-\mu\delta,
\end{equation}
where 
\begin{equation}
\rho_1=(1+\epsilon)\|I+D\|^2,~\rho_2=(1+\epsilon^{-1})(1+\xi)\|E\|^2,~\kappa=\varepsilon r, ~\textrm{and}~\mu=\lambda_{max}(A+A^T)+2\varepsilon.\nonumber
\end{equation}

Consider the following Lyapunov candidate $V(t)=V_1(t)+V_2(t)$ with
\begin{equation}
V_1(t)=x^Tx,~~V_2(t)=\varepsilon \int^0_{-r} x^T(t+s)x(t+s)\mathrm{d}s,
\end{equation}
then we can see that condition (i) of Theorem \ref{th3} is satisfied with $\alpha_1(\|x\|)=\alpha_2(\|x\|)=\|x\|^2$ and $\alpha_3(\|\phi\|_r)=\varepsilon r\|\phi\|^2_r$.

Considering the continuous dynamics of systems \eqref{sys1}, it follows
\begin{align}
\dot{V_1}(t) &= x^T(A+A^T)x + \sum^q_{i=1}w_i(2x^TA_ix +2x^TB_ix(t-r)) +2x^T Cw\cr
             &\leq x^T(A+A^T)x +q\bar{p}_1\|w\|\|x\|^2 + 2q\bar{p}_2\|w\|\|x\|\|x(t-r)\| +\varepsilon x^Tx +\varepsilon^{-1}\|C\|^2\|w\|^2\cr 
             &\leq x^T(A+A^T+\varepsilon I)x +(q\bar{p}_1\|w\|+\varepsilon^{-1} q^2\bar{p}^2_2\|w\|^2)\|x\|^2 + \varepsilon\|x(t-r)\|^2 +\varepsilon^{-1}\|C\|^2\|w\|^2\cr 
\dot{V_2}(t) &= \varepsilon \|x\|^2 -\varepsilon\|x(t-r)\|^2, \nonumber           
\end{align}
where
$$\bar{p}_1=\max_{i=1,..,q}\{\lambda_{max}(A_i+A^T_i)\} ~\textrm{and}~\bar{p}_2=\max_{i=1,..,q}\{\|B_i\|\},$$
then
\begin{align}\label{eqn4.2.2}
\dot{V}(t) &\leq x^T(A+A^T+2\varepsilon I)x +(q\bar{p}_1\|w\|+\varepsilon^{-1} q^2\bar{p}^2_2\|w\|^2)\|x\|^2 +\varepsilon^{-1}\|C\|^2\|w\|^2\cr
           &\leq (\chi_1(\|w\|)+\mu) V(t) +\chi_2(\|w\|)
\end{align}
where
$$\chi_1(\|w\|)=q\bar{p}_1\|w\|+\varepsilon^{-1} q^2\bar{p}^2_2\|w\|^2 ~\textrm{and}~ \chi_2(\|w\|)=\varepsilon^{-1}\|C\|^2\|w\|^2.$$

For $t=t_k$, we obtain from the impulses of system \eqref{sys1} that
\begin{align}\label{eqn4.2.3}
V_1(t_k) &= (D x(t^-_k) +E x(t_k-d)+ F w(t_k^-))^T(D x(t^-_k) +E x(t_k-d)+ F w(t_k^-))\cr
         &\leq \rho_1 V_1(t^-_k) +\rho_2 V_1(t_k-d) +\chi_3(\|w(t^-_k)\|),
\end{align}
where
$$\chi_3(\|w(t^-_k)\|)=(1+\epsilon^{-1})(1+\xi^{-1})\|F\|^2\|w(t^-_k)\|^2.$$
Moreover, we can verify
\begin{equation}\label{eqn4.2.4}
V_2(t)\leq \kappa \sup_{s\in[-r,0]}\{V_1(t+s)\}.
\end{equation}
We conclude from \eqref{eqn4.2.2},\eqref{eqn4.2.3},\eqref{eqn4.2.4} and \eqref{eqn4.2.1} that conditions (ii),(iii),(iv) and (v) of Theorem \ref{th3} hold. Therefore, system \eqref{sys1} is uniformly iISS over $\ell_{\textrm{sup}}(\delta)$.
\end{proof}

\begin{remark}
The iISS properties of system \eqref{sys1} without the impulse effects were initially studied in \cite{PP-ZPJ:2006}. Then the iISS results were extended to system \eqref{sys1} with switchings in \cite{JL-XL-WCX:2011}. But no time-delay was considered in the impulses. Theorem 3.10 in \cite{PP-ZPJ:2006} can be obtained from Theorem \ref{th4.1} when matrices $D,E$ and $F$ are zeros, and Theorem \ref{th4.1} and Theorem \ref{th4.2} reduce to Proposition 4.1 in \cite{JL-XL-WCX:2011} for system \eqref{sys1} without switchings when $E=0$. Our results are more general in the sense that we have generalized the iISS results in \cite{JL-XL-WCX:2011,PP-ZPJ:2006} to impulsive system \eqref{sys1} in which time-delay effects are considered in the impulses.
\end{remark}

\section{Illustrative Examples}\label{Sec5}
In this section, two numerical examples of system \eqref{sys1} are presented to illustrate the previous iISS results. 
\begin{example}\label{eg1}
To demonstrate Theorem \ref{th4.1}, we consider scalar bilinear system \eqref{sys1} with $x\in\mathbb{R}$, $w\in \mathbb{R}^2$ with $w_1(t)=t^{-2}$ and $w_2(t)=e^{-2t}$, $A=-1/2$, $A_1=1/2$, $A_2=1/4$, $B_1=1/3$, $B_2=1/5$, $C=[1/2~1/2]$, $F=[1/3~1/3]$, and $r=d=2/5$. 
\end{example}

{Next, we consider three types of impulses in system \eqref{sys1}.}
\begin{itemize}
\item $D=1/4$ and $E=1/5$.\\
Then $P=1$ so that $AP+PA=-1$ and
$$a=\frac{(1+D)P(1+D)}{P}=\frac{25}{16},~b=\frac{EPE}{P}=\frac{1}{25},~\mu_0=\min\Big\{ \frac{1}{P},\frac{1}{2r} \Big\}=1.$$
Hence, $\sqrt{a}+\sqrt{b}=1.45>1$ and Theorem \ref{th4.1} implies that system \eqref{sys1} is uniformly iISS over $\ell_{\textrm{inf}}(\delta)$ with $\delta$ satisfying
$$\delta>\frac{2\ln(\sqrt{a}+\sqrt{be^{\mu_0 r}})}{\mu_0}=0.8033.$$
As a numerical example, we take $t_{k+1}-t_k=\delta=1$ for all $k\in\mathbb{N}$. Simulation results for system \eqref{sys1} with the above given parameters are shown in Fig. \ref{fig1a} and \ref{fig1b}.
\begin{figure}[!t]
\centering
\subfigure[System response with $D=1/4$ and $E=1/5$.]{\label{fig1a}\includegraphics[trim={0 10cm 0 9cm},width=3.4in]{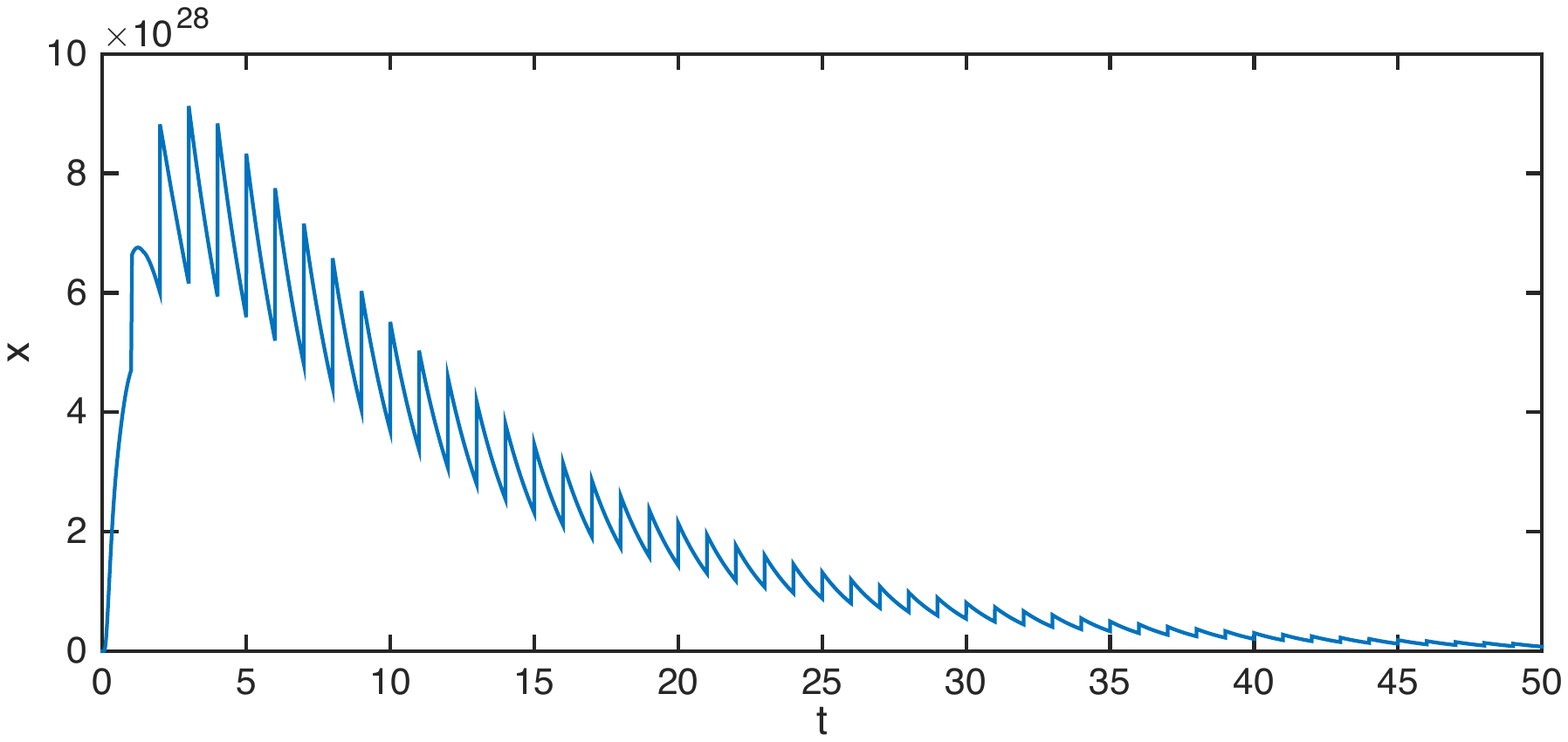}}
\subfigure[System response with $w=0$, $D=1/4$ and $E=1/5$.]{\label{fig1b}\includegraphics[trim={0 10cm 0 9cm},width=3.4in]{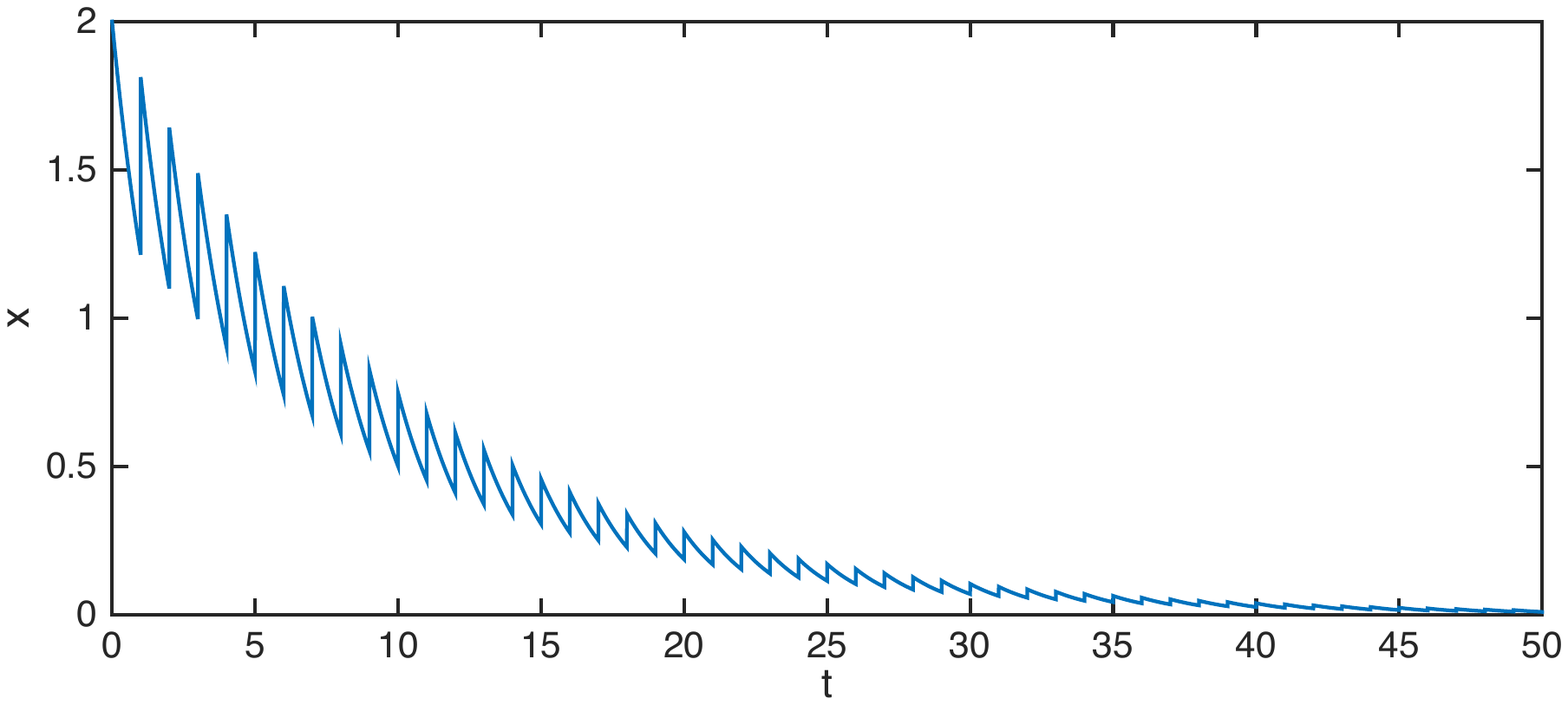}}
\subfigure[System response with $D=-1$ and $E=3/5$.]{\label{fig1'a}\includegraphics[trim={0 10cm 0 9cm},width=3.4in]{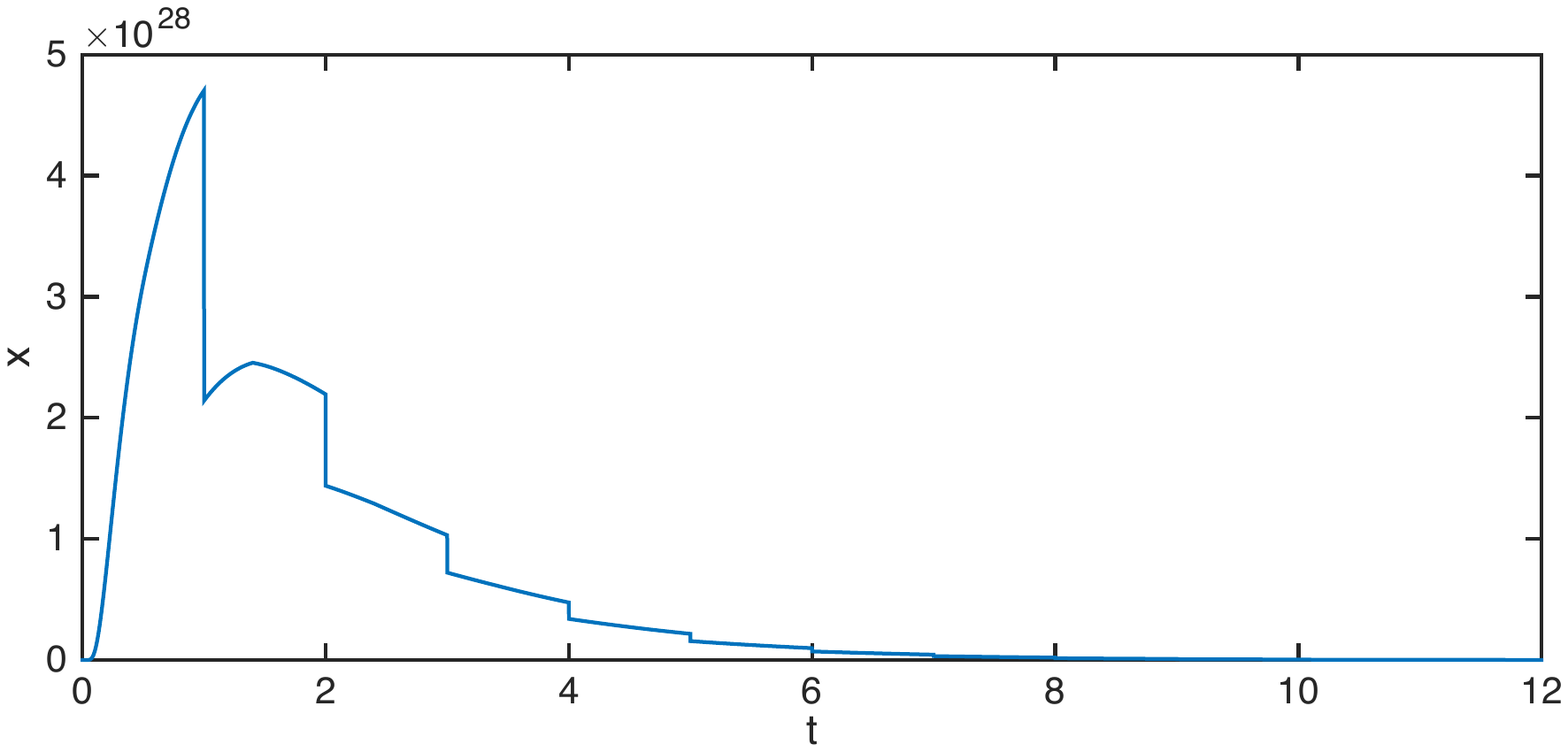}}
\subfigure[System response with $w=0$, $D=-1$ and $E=3/5$.]{\label{fig1'a1}\includegraphics[trim={0 10cm 0 9cm},width=3.4in]{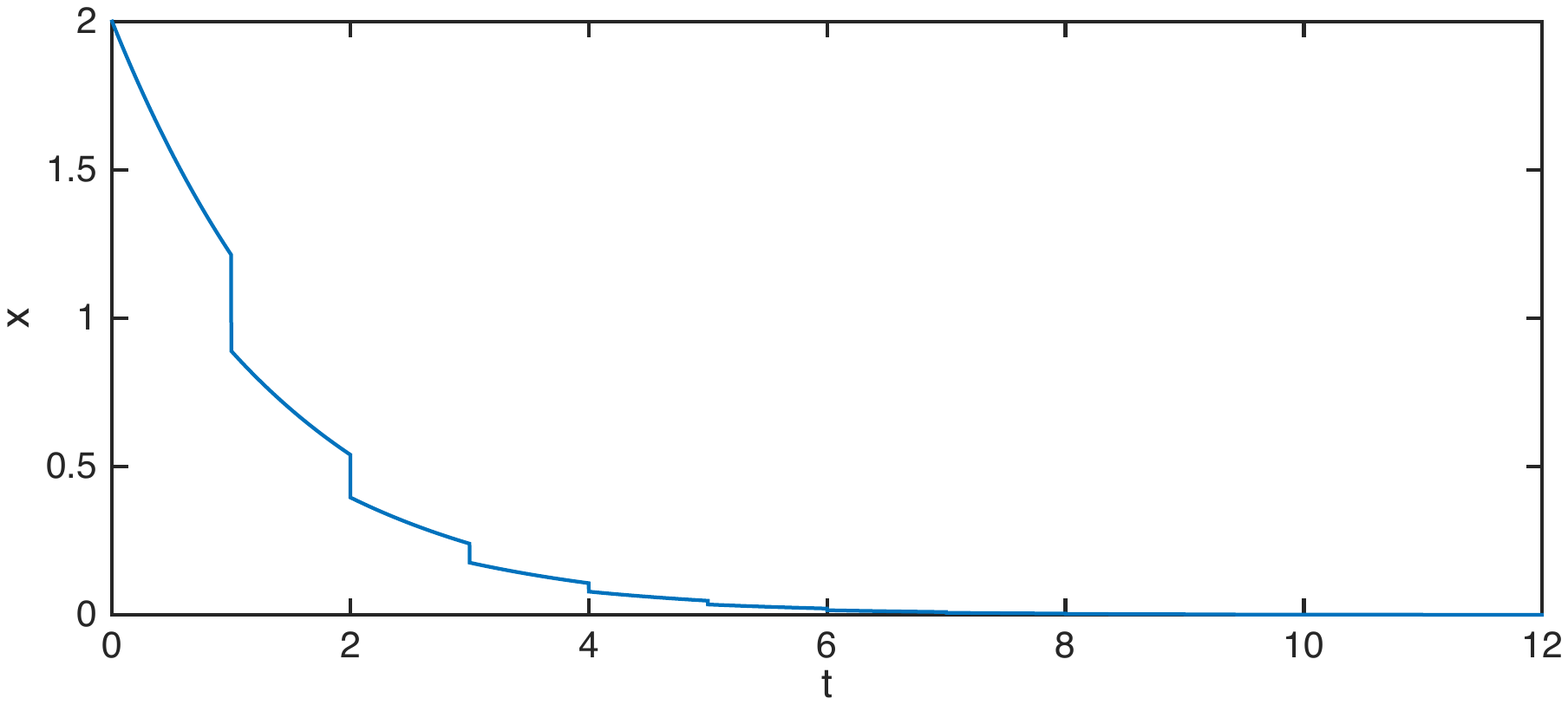}}
\subfigure[System response with $D=-1$ and $E=4/5$.]{\label{fig1'b}\includegraphics[trim={0 10cm 0 9cm},width=3.4in]{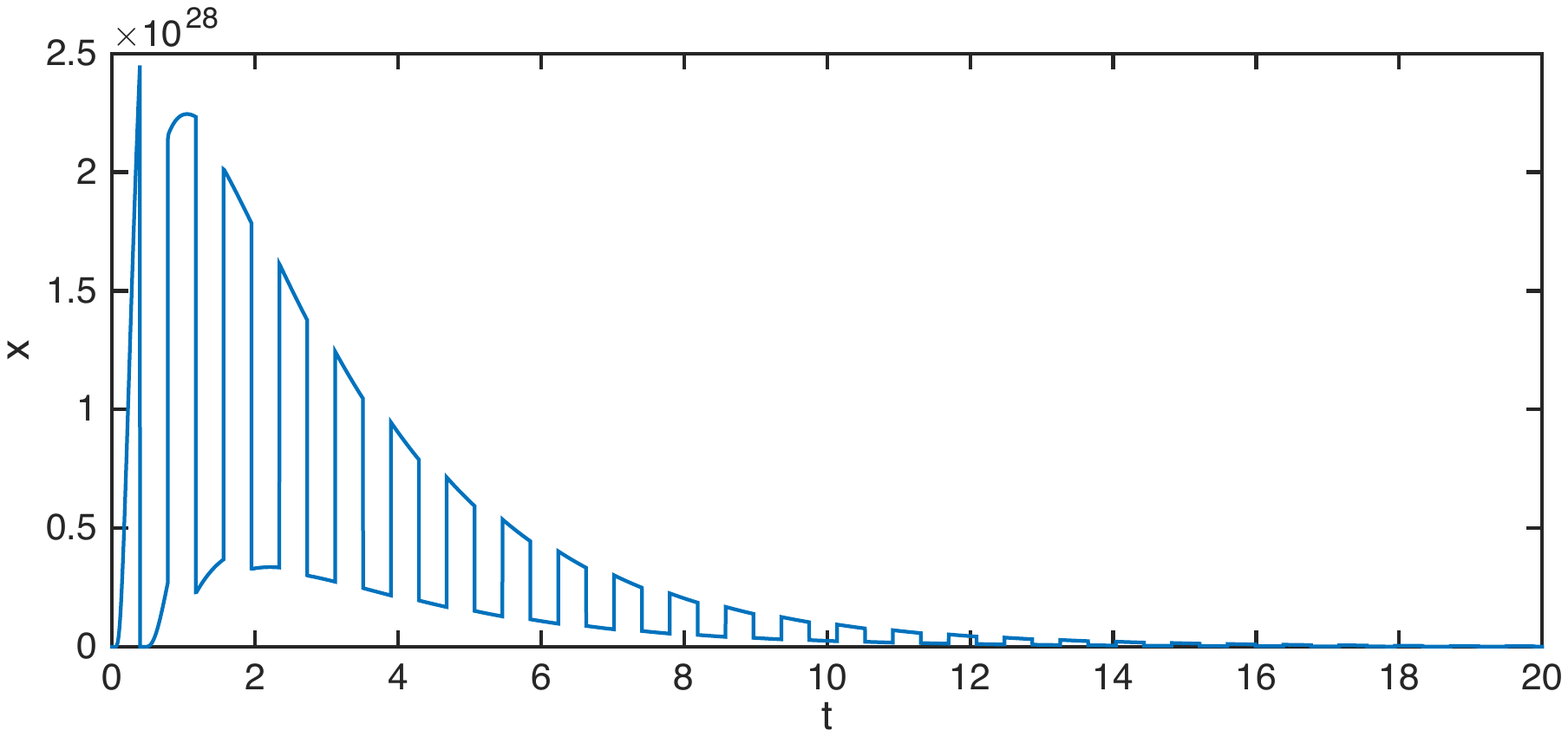}}
\subfigure[System response with $w=0$, $D=-1$ and $E=4/5$.]{\label{fig1'b1}\includegraphics[trim={0 10cm 0 9cm},width=3.4in]{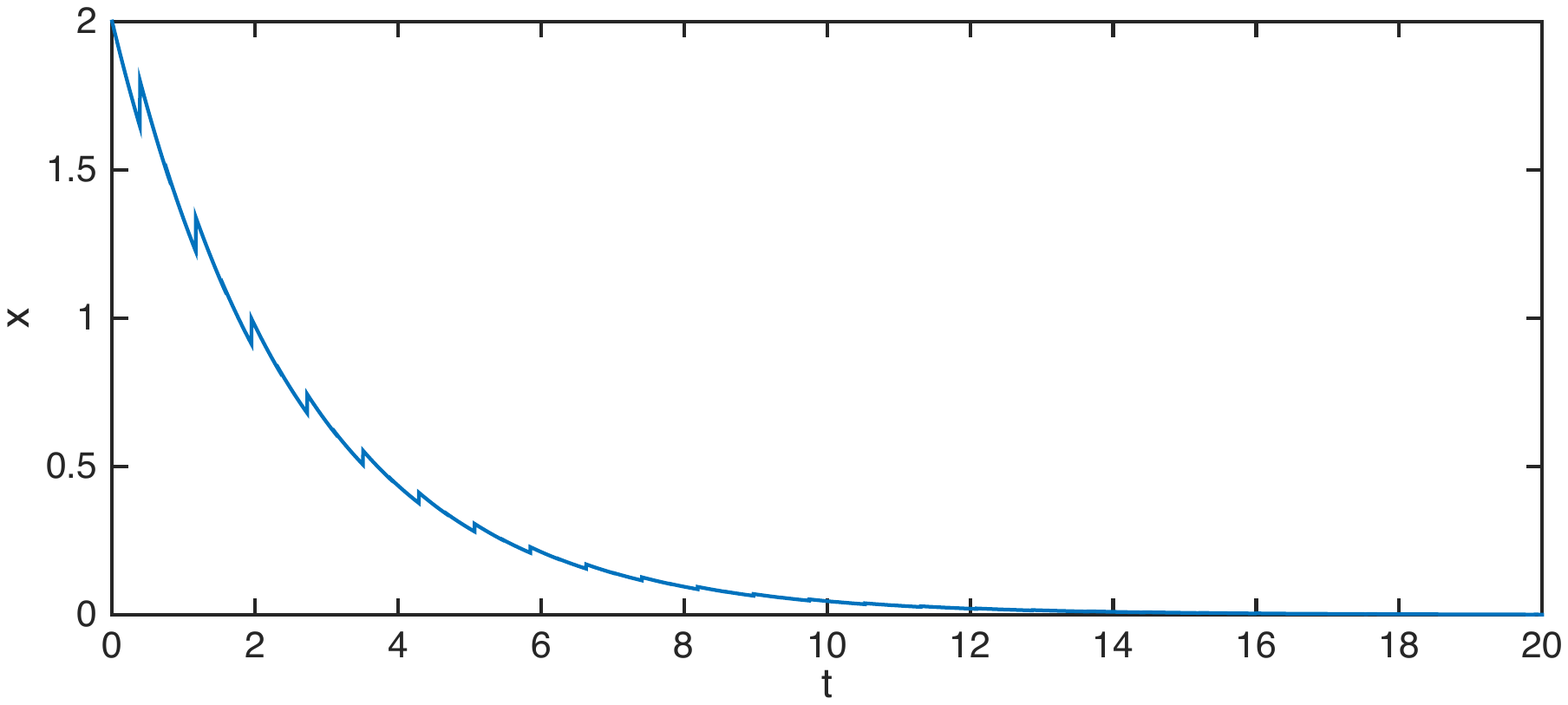}}
\caption{Simulation results for Example \ref{eg1} with initial condition $x_{t_0}(s)=2$ for $s\in[-r,0]$. System responses with the external inputs are shown in Fig. \ref{fig1a},\ref{fig1'a},\ref{fig1'b}. State trajectories of system \eqref{sys1} without the external inputs are given in Fig. \ref{fig1b},\ref{fig1'a1},\ref{fig1'b1}.}
\label{fig1}
\end{figure}
{
\item $D=-1$ and $E=3/5$.\\
With $P=1$, we have $a=(1+D)^2=0$, $b=E^2=9/25$ and $\mu_0=1$, which imply $\sqrt{a}+\sqrt{b}=3/5<1$. We then can conclude from Theorem \ref{th4.1} that system \eqref{sys1} is uniformly iISS over $\ell_{\textrm{all}}$. Simulation results with $t_{k+1}-t_k=\delta=1$ for all $k\in\mathbb{N}$ are given in Fig. \ref{fig1'a} and \ref{fig1'a1}. It can be seen that the delay-dependent impulses play a positive role in stabilizing the entire system.

\item $D=-1$ and $E=4/5$.\\
Similarly, let $P=1$ and then $a=0$, $b=16/25$ and $\mu_0=1$. Therefore, $\sqrt{a}+\sqrt{b}=4/5<1$ implies $\rho_1+\rho_2<1$ in Theorem \ref{th4}, and Theorem \ref{th4.1} tells that system \eqref{sys1} is uniformly iISS over $\ell_{\textrm{all}}$. See Fig. \ref{fig1'b} and \ref{fig1'b1} for numerical simulations with $t_{k+1}-t_k=\delta=0.39$ for all $k\in\mathbb{N}$. Fig. \ref{fig1'b} shows that some impulses are stabilizing while the others are destabilizing due to the existence of time-delay in the impulses. But all the impulses are destabilizing in Fig. \ref{fig1'b1}. This verified our discussion on the role of $\rho_2$ in Remark \ref{remark3}.
}
\end{itemize}

The following example is provided to show the effectiveness of Theorem \ref{th4.2}.
\begin{example}\label{eg2}
Consider bilinear system \eqref{sys1} with $x,w\in \mathbb{R}^2$, $A_1=0.5I$, $A_2=0.25I$, $B_1=I/3$, $B_2=0.2I$, $C=[0.5~0.5]$, $D=-0.65I$, $E=0.2I$, $F=[1/3~1/3]$, $r=d=0.4$, and
\[
A=\begin{bmatrix}
1.5 & 1 \\
0.5 & 2
\end{bmatrix}.
\]
\end{example}
It can be seen that $A$ is not Hurwitz. We conclude from Theorem \ref{th4.2} that system \eqref{sys1} with these parameters is uniformly iISS over $\ell_{\textrm{sup}}(\delta)$ with
$$\delta< \frac{2\ln(\|I+D\|+\|E\|)}{-\lambda_{max}(A+A^T)}=0.2011.$$
In the simulations, we take $t_{k+1}-t_k=\delta=0.2$ for all $k\in\mathbb{N}$ and the external input $w$ the same as that in Example \ref{eg1}. Simulation results for Example \ref{eg2} are shown in Fig. \ref{fig2}.
\begin{figure}[!t]
\centering
\subfigure[System response with the external input]{\label{fig2a}\includegraphics[width=3.4in]{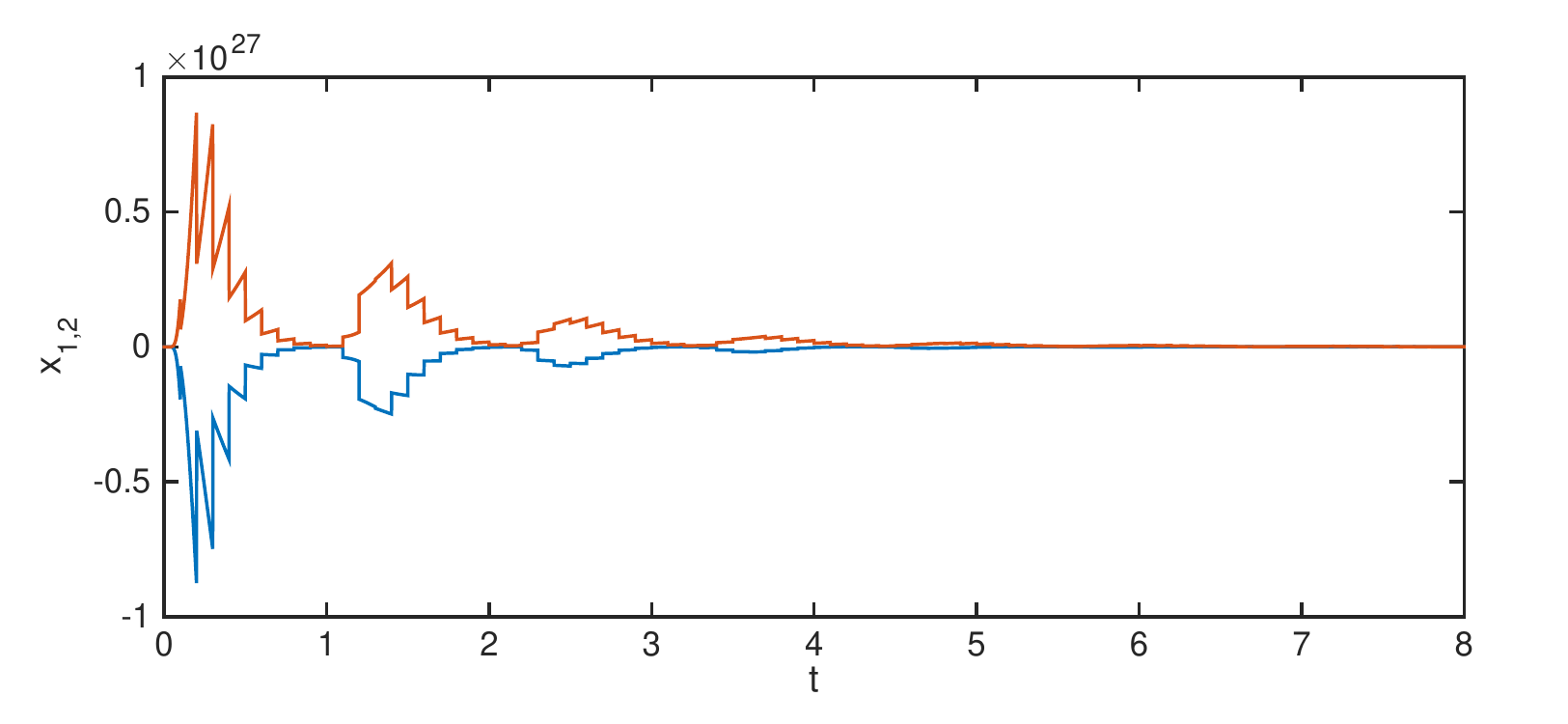}}
\subfigure[System response without the external input]{\label{fig2b}\includegraphics[width=3.4in]{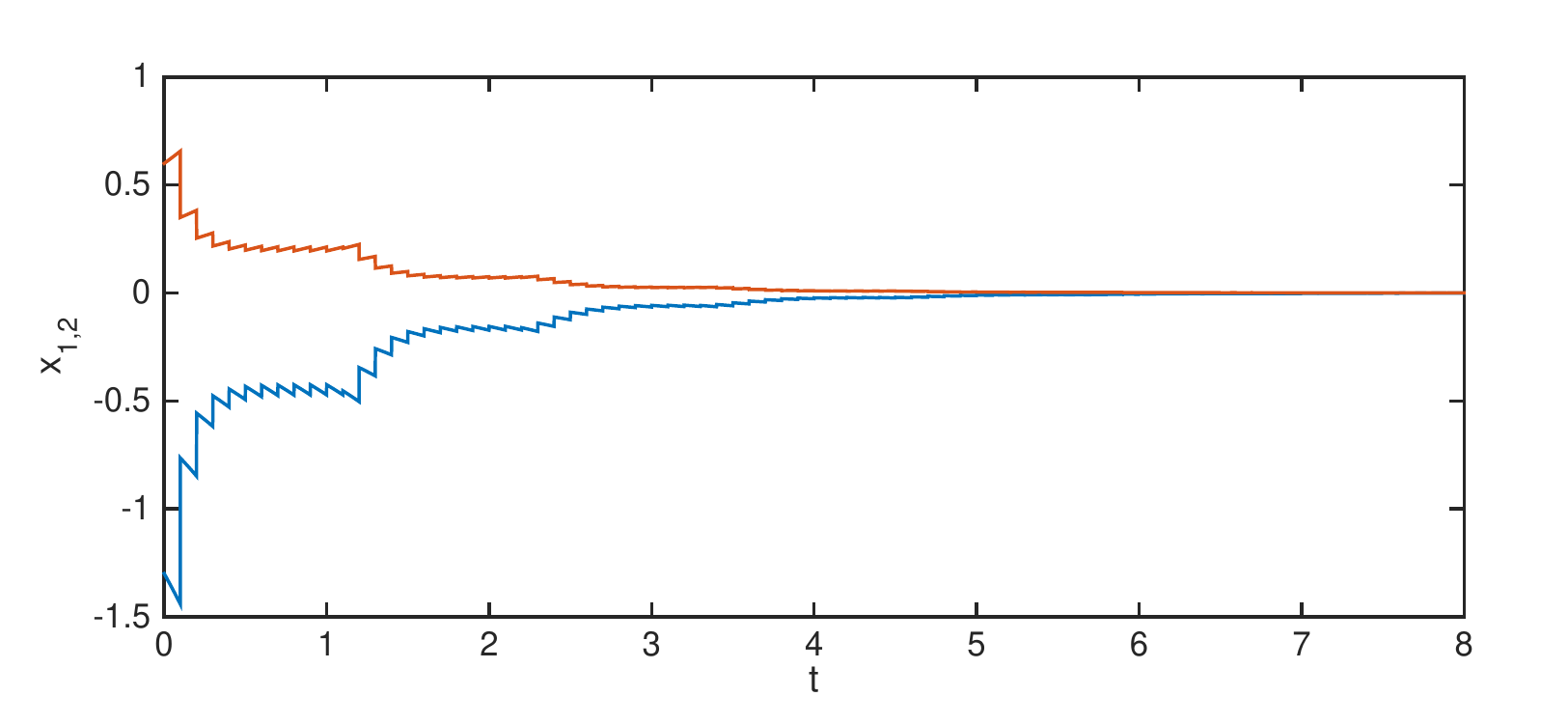}}
\caption{Simulation results for Example \ref{eg2} with $x_{t_0}(s)=[0.6~-1.4]^T$ for $s\in[-r,0]$.}
\label{fig2}
\end{figure}

%\begin{figure}[!t]
%\centering
%\subfigure[]{\label{fig2a}\includegraphics[width=3.4in]{fig02a}}
%\subfigure[]{\label{fig2b}\includegraphics[width=3.4in]{fig02b}}
%\subfigure[]{\label{fig2c}\includegraphics[width=3.4in]{fig02c}}
%\caption{Simulation results for Example \ref{example2} with (a) zero input $w(t)=0$; (b) exponentially decaying input $w(t)=e^{-7t}$; (c) periodic input $w(t)=\cos(-16\pi t)$.}
%\label{fig2}
%\end{figure}

\section{Conclusions}\label{Sec6}
The method of Lyapunov-Krasovskii functionals has been used to study the iISS properties of impulsive systems with time-delay. The main contribution of this paper is that iISS criteria have been obtained for general nonlinear impulsive systems with time-delay in both the continuous dynamics and the impulses. The iISS properties of a class of bilinear systems have been investigated in great details to demonstrate the effectiveness of our iISS results. Numerical simulations of two illustrative examples have been provided accordingly. An interesting topic of the current research is to study the event-triggered control applications of the obtained iISS results (see, e.g. \cite{HY-FH-XC:2019}).

\section*{Acknowledgments}
The author would like to thank the anonymous reviewers, whose constructive comments and suggestions have improved the quality of this paper. This research was partially supported by the Coleman Postdoctoral Fellowship from Queen's University at Kingston, which is gratefully acknowledged.

%\section*{References}
{\footnotesize
%\bibliographystyle{plain}{}

}

\end{document}